\documentclass[12pt]{article}

\newcommand{\bbsm}{ \begin{smallmatrix}}
	\newcommand{\besm}{\end{smallmatrix}}

\usepackage[hmargin=1.3in,vmargin=1.3in]{geometry}
\usepackage[hmargin=1.3in,vmargin=1.3in]{geometry}
\usepackage{bbding}
\usepackage{mathrsfs}
\usepackage{cite}
\usepackage{amsfonts}
\usepackage{multirow}
\usepackage{amsfonts,amssymb,amsmath,amsthm,bm}
\usepackage[colorlinks,
            linkcolor=blue,
            anchorcolor=blue,
            citecolor=blue]{hyperref}

\newtheorem{theorem}{Theorem}
\newtheorem{example}{Example}
\newtheorem{definition}{Definition}
\newtheorem{remark}{Remark}
\newtheorem{corollary}{Corollary}
\newtheorem{lemma}{Lemma}

\begin{document}

\title{A Note on Self-Dual Generalized Reed-Solomon Codes}
\author{\small Weijun Fang $^{1,2}$ \ Jun Zhang $^{3,}$\footnote{Corresponding Author} \ Shu-Tao Xia$^{1,2}$ \ Fang-Wei Fu$^{4,5}$ \\
\scriptsize $^1$ Tsinghua Shenzhen International Graduate School, Tsinghua University, Shenzhen 518055,  China \\
\scriptsize $^2$ PCL Research Center of Networks and Communications, Peng Cheng Laboratory, Shenzhen 518055,  China\\
\scriptsize $^3$ School of Mathematical Sciences, Capital Normal University, Beijing 100048, China\\
\scriptsize $^4$  Chern Institute of Mathematics and LPMC, Nankai University, Tianjin 300071,  China\\
\scriptsize $^5$ Tianjin Key Laboratory of Network and Data Security Technology, Nankai University, Tianjin 300071,  China\\
\scriptsize E-mail: nankaifwj@163.com, junz@cnu.edu.cn, xiast@sz.tsinghua.edu.cn, fwfu@nankai.edu.cn\\
}
\date{}
\maketitle
\thispagestyle{empty}
\begin{abstract}
A linear code is called an MDS self-dual code if it is both an MDS code and a self-dual code with respect to the Euclidean inner product. The parameters of such codes are completely determined by the code length. In this paper, we consider new constructions of MDS self-dual codes via generalized Reed-Solomon (GRS) codes and their extended codes.  The critical idea of our constructions is to choose suitable evaluation points such that the corresponding (extended) GRS codes are self-dual. The evaluation set of our constructions is consists of a subgroup of finite fields and its cosets in a bigger subgroup. Four new families of MDS self-dual codes are obtained and they have better parameters than previous results in certain region. Moreover, by the M{\"o}bius action over finite fields, we give a systematic way to construct self-dual GRS codes with different evaluation points provided any known self-dual GRS codes. Specially, we prove that all the self-dual extended GRS codes over $\mathbb{F}_{q}$ with length $n< q+1$ can be constructed from GRS codes with the same parameters.

\end{abstract}

\small\textbf{Keywords:} MDS codes, Self-dual codes, Generalized Reed-Solomon codes, Mobius action

\maketitle

\section{Introduction}

An $[n, k, d]$-linear code $C$ over $\mathbb{F}_{q}$ is just a subspace of $\mathbb{F}_{q}^{n}$ with dimension $k$ and minimum Hamming distance $d$, where $d$ is given as
\[d \triangleq \min_{\bm c =(c_{1}, \dots, c_{n}) \neq 0 \in C}|\{1 \leq i \leq n: c_{i} \neq 0 \}|.\]
The code rate defined by $k/n$ and the minimum Hamming distance $d$ are two important metrics of the code. One of the trade-off between these parameters is the well-known Singleton bound:
\[d \leq n-k+1.\]
An $[n, k, d]$-linear code $C$ achieving the Singleton bound is called a \emph{maximum distance separable} (MDS) code. MDS codes have various interesting properties and wide applications both in theoretical and practice.  On the one hand, MDS codes are closely related to many other mathematical aspects, such as orthogonal arrays in combinatorial design and $n$-arcs in finite geometry (cf. \cite[Chap. 11]{MS77}). On the other hand, MDS codes also have widespread applications in data storage, such as CD ROMs and coding for distributed storage systems (see \cite{overview}). In particular, the most commonly used class of MDS codes is the well-known Reed-Solomon codes, which not only have nice theoretic properties, but also have been extensively applied in engineering due to their easy encoding and efficient decoding algorithm.

The Euclidean inner product $\langle\cdot,\cdot\rangle$ in $\mathbb{F}_{q}^{n}$ is given as
\[\langle \textbf{u}, \textbf{v}\rangle = \sum_{i=1}^{n} u_i v_i,\]
where $\textbf{u}= (u_{1}, \ldots, u_{n}) \in \mathbb{F}_{q}^{n}$ and $\textbf{v}= (v_{1}, \ldots, v_{n}) \in \mathbb{F}_{q}^{n}$.
 For any $[n, k, d]_{q}$-linear code $C$, we define the dual code $C^{\perp}$ of $C$ as
\[C^{\perp} := \{\textbf{u} \in \mathbb{F}_{q}^{n}: \langle \textbf{u}, \textbf{c}\rangle=0, \textnormal{ for any } \textbf{c} \in C \}.\]
One of important connections between $C$ and its dual $C^{\perp}$ is the well-known MacWilliams identity (cf.\cite{MS77}). We say $C$ is a \emph{self-dual} code if $C^{\perp}=C$. Self-dual codes is a very interesting and important class of linear codes. The length of a self-dual code is obviously even and the dimension is equal to half of length. The existence result of self-dual codes was given by Pless in \cite{P68}, which showed that a $q$-ary self-dual code of even length $n$ exists if and only if $(-1)^{\frac{n}{2}}$ is a square element in $\mathbb{F}_{q}$. Self-dual codes have also some applications in other aspects, such as linear secret sharing schemes (see \cite{C08,DMS08}) and unimodular integer lattices (see \cite{HP03, CS99}). In past decades, it attracts lots of attentions for investigating the codes which are both MDS and self-dual. Such codes are called MDS self-dual codes. Note that the dimension and minimum distance of an MDS self-dual code over $\mathbb{F}_{q}$ are equal to  $\frac{n}{2}$ and  $\frac{n}{2}+1$, respectively. Therefore, a natural question is that for which even $n$ an MDS self-dual code over $\mathbb{F}_{q}$ of length $n$ exists. According to the well-known MDS conjecture, it is expected to determine the existence of $q$-ary MDS self-dual codes of length $n$ for all possible $n \leq q+1$.

Recent years, several progress have been made in this topic (see \cite{GK02,BGGHK03,BBDW04,KL04,HL06,GKL08,YC15}). For the case where $q$ is even, it was shown by Grassl and Gulliver in \cite{GG08} that there is an  MDS self-dual code of even length $n$ for all $n \leq q$. By using the properties of cyclic and constacyclic codes, some new MDS self-dual codes were obtained in \cite{G12} and \cite{TW17}. A systematic construction of MDS self-dual codes was first proposed by Jin and Xing in \cite{JX17} via generalized Reed-Solomon (GRS) codes over finite fields. A subset of elements of $\mathbb{F}_{q}$, called evaluation set, with some special properties can be used to construct desired GRS codes which are self-dual. Since then, GRS codes becomes one of the most commonly used tools to construct MDS self-dual codes. In \cite{Y18} and \cite{FF19}, the authors generalized this method to extended GRS codes with general length. The evaluation set with special structures, such as a multiplicative subgroup of $\mathbb{F}_{q}^{*}$, a subspace of $\mathbb{F}_{q}$ and their cosets, are considered for constructing self-dual GRS codes. Zhang and Feng \cite{ZF19,ZF19-2} presented a unified approach to MDS self-dual codes and obtain some new codes via cyclotomy. Fang \emph{et al}. \cite{FLLL19,LLL19,FLL19} constructed new families of self-dual GRS codes via two disjoint multiplicative subgroups of $\mathbb{F}_{q}^{*}$ and their cosets. In \cite{S19}, by using the technique of algebraic geometry codes, Sok constructed several new families of MDS self-dual codes.  In the following Table 1, we summary some known results about the construction of MDS self-dual codes.

\newcommand{\tabincell}[2]{\begin{tabular}{@{}#1@{}}#2\end{tabular}}
\begin{table}[htbp]
\footnotesize
\centering
\caption{Some known results on MDS self-dual codes of even length $n$}
\vskip 3mm
\begin{tabular}{|c|c|c|}
  \hline
  $q$ & $n$ & References \\
  \hline
  $q$ even  & $n \leq q$ & \cite{GG08} \\
  \hline
  $q$ odd  & $n = q+1$ & \cite{GG08,JX17} \\
  \hline
  $q=r^{2}$ & $n \leq r$ & \cite{JX17} \\
  \hline
  \tabincell{c}{$q=r^{2}$,\\ $r \equiv 3 (\textnormal{mod } 4)$} & $n=2tr, t \leq \frac{r-1}{2}$ & \cite{JX17} \\
  \hline
  $q \equiv 1 (\textnormal{mod } 4)$ & $4^{n}n^{2} \leq q$ & \cite{JX17} \\
  \hline
  $q \equiv \textnormal{3 (mod 4)}$ & $n \equiv \textnormal{0 (mod 4)}$ and $(n-1) \mid (q-1)$ & \cite{TW17}\\
  \hline
  $q \equiv 1 (\textnormal{mod } 4)$ & $(n-1) \mid (q-1)$ & \cite{TW17} \\
  \hline
  $q \equiv 1 (\textnormal{mod } 4)$ & $n=2p^{\ell}, \ell \leq m$ & \cite{FF19} \\
  \hline
  $q \equiv 1 (\textnormal{mod } 4)$ & $n=p^{\ell}+1, \ell \leq m$ & \cite{FF19} \\
  \hline
  \tabincell{c}{$q=r^{s}$,\\ $r$ odd, $s$ even} & \tabincell{c}{$n=2tr^{\ell}$, $0 \leq \ell <s$, \\ and $1 \leq t \leq \frac{r-1}{2}$} & \cite{FF19} \\
  \hline
  \tabincell{c}{$q=r^{s}$, \\$r$ odd, $s$ is even} & \tabincell{c}{$n=(2t+1)r^{\ell}+1$, $0 \leq \ell <s$, \\and $0 \leq t \leq \frac{r-1}{2}$} & \cite{FF19} \\
  \hline
  $q$ odd & $(n-2) \mid (q-1)$, $\eta(2-n)=1$ & \cite{FF19,Y18} \\
  \hline
  $q \equiv 1 (\textnormal{mod } 4)$ & $n \mid (q-1)$ & \cite{Y18} \\
  \hline
  \tabincell{c}{$q=r^{s}$,\\ $r$ odd, $s \geq 2$ }& $n=tr$, $t$ even  and  $2t \mid (r-1)$ &\cite{Y18}\\
  \hline
  \tabincell{c}{$q=r^{s}$, \\$r$ odd, $s \geq 2$} & \tabincell{c}{$n=tr$, $t$ even, $(t-1) \mid (r-1)$ \\and $\eta(1-t)=1$} &\cite{Y18}\\
  \hline
  \tabincell{c}{$q=r^{s}$, \\$r$ odd, $s \geq 2$} & \tabincell{c}{$n=tr+1$, $t$ odd, $t \mid (r-1)$ \\and $\eta(t)=1$} &\cite{Y18}\\
  \hline
  \tabincell{c}{$q=r^{s}$,\\ $r$ odd, $s \geq 2$} & \tabincell{c}{$n=tr+1$, $t$ odd, $(t-1) \mid (r-1)$\\ and $\eta(t-1)=\eta(-1)=1$} &\cite{Y18}\\
  \hline
  $q=r^{2}$, $r$ odd & \tabincell{c}{$n=tm$, $\frac{q-1}{m}$ even, \\ and $1 \leq t \leq \frac{r+1}{\gcd(r+1, m)}$} & \cite{FLLL19}\\
  \hline
  $q=r^{2}$, $r$ odd & \tabincell{c}{$n=tm+1$, $tm$ odd, $m \mid (q-1)$, \\ and $2 \leq t \leq \frac{r+1}{2\gcd(r+1, m)}$} & \cite{FLLL19}\\
  \hline
  $q=r^{2}$, $r$ odd & \tabincell{c}{$n=tm+2$, $m \mid (q-1)$, $tm$ even \\(except $t, m$ are even and $r \equiv 1 (\textnormal{mod } 4)$),  \\ and $1 \leq t \leq \frac{r+1}{\gcd(r+1, m)}$} & \cite{FLLL19}\\
  \hline
  $q=r^{2}$, $r$ odd & \tabincell{c}{$n=tm$,  $1 \leq t \leq \frac{s(r-1)}{\gcd(s(r-1), m)}$\\$\frac{q-1}{m}$ even, $s$ even, $s \mid m$, and $\frac{r+1}{s}$ even } & \cite{FLLL19}\\
  \hline
  $q=r^{2}$, $r$ odd & \tabincell{c}{$n=tm+2$,  $1 \leq t \leq \frac{s(r-1)}{\gcd(s(r-1), m)}$\\ $\frac{q-1}{m}$ even, $s$ even, $s \mid m$, and $\frac{r+1}{s}$ even } & \cite{FLLL19}\\
  \hline
  $q=r^{2}$, $r$ odd & \tabincell{c}{$n=tm$, $\frac{q-1}{m}$ even, \\ and $1 \leq t \leq \frac{r-1}{\gcd(r-1, m)}$} & \cite{LLL19}\\
  \hline
  $q=r^{2}$, $r$ odd & \tabincell{c}{$n=tm+1$, $tm$ odd, $m \mid (q-1)$, \\ and $2 \leq t \leq \frac{r-1}{\gcd(r-1, m)}$} & \cite{LLL19}\\
  \hline
  $q=r^{2}$, $r$ odd & \tabincell{c}{$n=tm+2$, $tm$ even, $m \mid (q-1)$, \\ and $2 \leq t \leq \frac{r-1}{\gcd(r-1, m)}$} & \cite{LLL19}\\
  \hline
  \tabincell{c}{$q=r^{2}$,\\ $r \equiv 1 (\textnormal{mod } 4)$ }& \tabincell{c}{$n=s(r-1)+t(r+1)$, $s$ even, \\$1 \leq s \leq \frac{r+1}{2}$ $1 \leq t \leq \frac{r-1}{2}$} & \cite{FLL19} \\
  \hline
  \tabincell{c}{$q=r^{2}$, \\$r \equiv 3 (\textnormal{mod } 4)$} & \tabincell{c}{$n=s(r-1)+t(r+1)$, $s$ odd, \\$1 \leq s \leq \frac{r+1}{2}$ $1 \leq t \leq \frac{r-1}{2}$} & \cite{FLL19} \\
  \hline
\end{tabular}
\end{table}

In this paper, we investigate the construction of MDS self-dual codes with new parameters by using (extended) GRS codes over finite fields. The main idea of our constructions is to consider two multiplicative subgroups $H$ and $G$ of $\mathbb{F}_{q}^{*}$, where $H$ is a subgroup of $G$. We consider the union of some costes of $H$ in $G$ as the evaluation set. Then we present two new constructions (see Theorems \ref{thm1} and \ref{thm2}) of MDS self-dual codes via Lemmas \ref{lem2} and \ref{lem3}.  Moreover, by considering some automorphisms of GRS codes, we give a systematic way to construct self-dual GRS codes provided any known self-dual GRS code. As a corollary, all the self-dual extended GRS codes over finite field $\mathbb{F}_{q}$ with length $n< q+1$ can be constructed from self-dual GRS codes with the same parameters.

We list the parameters of our new MDS self-dual codes as follows. Specifically, if one of the following conditions holds, then there exists a self-dual GRS code of length $N$ over $\mathbb{F}_{q}$, where $q=r^{2}$ and $r$ is a power of an odd prime $p$.

\begin{description}
  \item[(i)] $N=tn'$ is even, $n' \mid (q-1)$, $n_{2}=\frac{r+1}{\gcd(r+1, n')}$ is even, and $1 \leq t \leq \frac{r-1}{n_{2}}$; (see Theorem \ref{thm1} (i))
  \item[(ii)] $N=tn'+2$ is even, $n' \mid (q-1)$, $n_{2}=\frac{r+1}{\gcd(r+1, n')}$, and $1 \leq t \leq \frac{r-1}{n_{2}}$; (see Theorem \ref{thm1} (iii))
  \item[(iii)] $N=tn'$ is even, $n' \mid (q-1)$, both $\frac{r-1}{n_{1}}$ and $tn_{2}$ are even, where $n_{1}=\gcd(r-1, n')$, $n_{2}=\frac{r-1}{\gcd(r-1, n')}$ , and $1 \leq t \leq \frac{r-1}{n_{2}}$; (see Theorem \ref{thm2} (i))
  \item[(iv)] $N=tn'+2$ is even, $n' \mid (q-1)$, $n_{2}=\frac{r-1}{\gcd(r-1, n')}$ both $n_{2}$ and $\frac{r+1}{2}(t-1)$ are even, or $n_{2}$ is odd and $t$ is even with $1 \leq t \leq \frac{r-1}{n_{2}}$; (see Theorem \ref{thm2} (ii))
\end{description}

The rest of this paper is organized as follows. In Section 2, we briefly introduce some basic notations and results about generalized Reed-Solomon codes and MDS self-dual codes. In Section 3, four new families of self-dual GRS codes are constructed. In Section 4, we give a systematic way to construct self-dual GRS codes from known self-dual GRS code. Finally, we conclude this paper in Section 5.


\section{Preliminaries}
In this section, we recall some basic notations and results about generalized Reed-Solomon codes and MDS self-dual codes.

Let $q$ be a prime power and $\mathbb{F}_{q}$ be the finite field with $q$ elements. Let $n$ and $k$ be two integers such that $2\leq n\leq q+1$ and $2\leq k\leq n.$  Write the finite field $\mathbb{F}_q=\{\alpha_1, \cdots, \alpha_q\}$ and let $\infty$ be the infinity point. For any $\alpha\in \mathbb{F}_{q}\cup\infty$, let $\bm c_k(\alpha)$ be a column vector of $\mathbb{F}_{q}^k$ defined as
\[
\bm c_k(\alpha)\triangleq\begin{cases}
(1,\alpha,\alpha^2,\cdots, \alpha^{k-1})^T\in\mathbb{F}_{q}^k&\mbox{if $\alpha\in\mathbb{F}_{q}$;}\\
(0,0,\cdots,0,1)^T\in\mathbb{F}_{q}^k&\mbox{if $\alpha=\infty$.}\\
\end{cases}
\]
\begin{definition}
Let $A=\{a_1,a_2,\cdots, a_n\}$ be a subset of $\mathbb{F}_{q}\cup\infty$ and $\bm v=(v_1,v_2,\cdots,v_n)$ be a vector in $(\mathbb{F}_{q}^*)^n$. For any $2\leq k\leq n,$ if $\infty \notin A$, the linear code with generator matrix $[v_1\bm c_k(a_1), v_2\bm c_k(a_2), \cdots, v_n\bm c_k(a_n)]=[\bm c_k(a_1), \bm c_k(a_2), \cdots, \bm c_k(a_n)]\mbox{diag}\{v_1,v_2,\cdots,v_n\}$ is called the generalized Reed-Solomon code. If $\infty\in S$, then the linear code is called extended generalized Reed-Solomon code. We call them (extended) generalized Reed-Solomon (GRS) code with evaluation set $A$ and scaling vector $\bm v$, and denoted by $GRS_k(A,\bm v)$.
\end{definition}
 It is well-known that $GRS_k(A,\bm v)$ is an $[n,k,n-k+1]$-MDS code and its dual code is also a GRS code.

Recently, lots of research work on construction of MDS self-dual codes have been done by using GRS codes. The key point of these constructions is to choose suitable evaluation set $A$ such that the corresponding GRS code is self-dual for some scaling vector $\bm v=(v_1,v_2,\cdots,v_n)\in(\mathbb{F}_{q}^*)^n$. In the following, we introduce some related notations and results.

Let $A$ be a subset of  $\mathbb{F}_{q}$, we define the polynomial $\pi_{A}(x)$ over $\mathbb{F}_{q}$ as
\[\pi_{A}(x) \triangleq \prod_{a \in A}(x-a).\]
For any element $a \in A$, we define
\[\delta_{A}(a) \triangleq \prod_{a' \in A, a' \neq a}(a-a').\]

The properties of $\pi_{A}(x)$ and $\delta_{A}(a)$ are given as follows, which were first obtained in \cite{ZF19}. We present its proof for the completeness.
\begin{lemma}\cite[Lemma 3.1]{ZF19}\label{lem1}
\begin{description}
  \item[(i)] Let $A$ be a subset of $\mathbb{F}_{q}$, then for any $a \in A$
  \[\delta_{A}(a)=\pi'_{A}(a),\]
  where $\pi'_{A}(x)$ is the derivative of $\pi_{A}(x)$.
  \item[(ii)] Let $A_{1},A_{2},\dots, A_{m}$ be $m$ pairwise disjoint subsets of $\mathbb{F}_{q}$, and $A=\bigcup_{\ell=1}^{m}A_{\ell}$. Then for any $a \in A_{i}$,
      \[\delta_{A}(a)=\delta_{A_{i}}(a)\prod_{1 \leq j \leq \ell, j \neq i}\pi_{A_{j}}(a).\]
\end{description}
\end{lemma}

\begin{proof}
\begin{description}
  \item[(i)] The conclusion follows from $\pi'_{A}(x)=\sum_{a \in A}\prod_{a' \in A, a' \neq a}(x-a').$
  \item[(ii)] Note that $\pi_{A}(x)=\prod_{\ell=1}^{m}\pi_{A_{\ell}}(x)$. Thus
  \[\pi'_{A}(x)=\sum_{\ell=1}^{m}\pi'_{A_{\ell}}(x)\prod_{1 \leq j \leq m, j \neq \ell}\pi_{A_{j}}(x).\]
  Part (ii) then follows from Part (i) and the fact that $\pi_{A_{i}}(a)=0$ for any $a \in A_{i}$.
\end{description}

\end{proof}
Let $\eta(x)$ be the quadratic character of $\mathbb{F}^{*}_{q}\triangleq\mathbb{F}_{q}\backslash\{0\}$, that is $\eta(x)=1$ if $x$ is a square in $\mathbb{F}^{*}_{q}$ and $\eta(x)=-1$ if $x$ is a non-square in $\mathbb{F}^{*}_{q}$. The following Lemmas \ref{lem2} and \ref{lem3} are the key lemmas for the construction of MDS self-dual codes, which have been used in the literature with some equivalent forms. The reader may refer to \cite{JX17, Y18, FF19, ZF19, FLLL19, LLL19, FLL19} for more details on their proofs.
\begin{lemma}\label{lem2}
Suppose $n$ is even. Let $A$ be a subset of $\mathbb{F}_{q}$ of size $n$, such that for all $a \in A$, $\eta\big(\delta_{A}(a)\big)$ are the same. Then there exists a vector $\bm v \in(\mathbb{F}_{q}^*)^n$ such that $GRS_{\frac{n}{2}}(A,\bm v)$ is self-dual. Consequently, there exists a $q$-ary MDS self-dual code of length $n$.
\end{lemma}

\begin{lemma}\label{lem3}
Suppose $n$ is odd. Let $A$ be a subset of $\mathbb{F}_{q}$ of size $n$, such that for all $a \in A$, $\eta\big(-\delta_{A}(a)\big)=1$. Then there exists a vector $\bm v \in(\mathbb{F}_{q}^*)^{n+1}$ such that $GRS_{\frac{n+1}{2}}(A\cup \infty, \bm v)$ is self-dual. Consequently, there exists a $q$-ary MDS self-dual code of length $n+1$.
\end{lemma}

\section{New Constructions of MDS Self-Dual Codes}

In this section, based on Lemmas \ref{lem2} and \ref{lem3}, we give two new constructions of MDS self-dual codes via different multiplicative subgroups of finite fields. Throughout this section, we suppose that $q=r^{2}$ and $r=p^{m}$, where $p$ is an odd prime.

 Let $n'$ be a positive integer with $n' \mid (q-1)$. We write $n'=n_{1}n_{2}$, where $n_{1}=\gcd(n', r+1)$ and $n_{2}=\frac{n'}{n_{1}}=\frac{n'}{\gcd(n',r+1)}$. Then $n_{2} \mid (r-1)\frac{r+1}{n_{1}}$. Note that $\gcd(n_{2}, \frac{r+1}{n_{1}})=\gcd(\frac{n'}{n_{1}}, \frac{r+1}{n_{1}})=1$, hence
\[n_{2} \mid (r-1).\]
Let $\omega$ be a primitive element of $\mathbb{F}_{q}$. Denote
\[H=\langle \omega^{\frac{q-1}{n'}} \rangle, G=\langle\omega^{\frac{r+1}{n_{1}}}\rangle.\]
Then $|H|=n'$. Note that $\frac{q-1}{n'}=\frac{r+1}{n_{1}}\cdot\frac{r-1}{n_{2}}$, thus $ \frac{r+1}{n_{1}} \mid \frac{r-1}{n'}$, which deduces that $H$ is a subgroup of $G$.
Then there exist $\beta_{1}, \ldots , \beta_{\frac{r-1}{n_{2}}} \in G$
such that $\{\beta_{b} H\}^{\frac{r-1}{n_{2}}}_{b
=1}$ represent all cosets of $G/H$.

Let $1 \leq t \leq \frac{r-1}{n_{2}}$ and $n=tn'$. Denote $A_{b}=\beta_{b}H$,
\begin{equation}\label{1}
 A=\bigcup^{t}_{b=1}A_{b},
\end{equation}
and
\begin{equation}\label{2}
 A_{0}=A\cup\{0\}.
\end{equation}

Based on \eqref{1} and \eqref{2}, we give our first construction as follows.

\begin{theorem}\label{thm1}
   Let $q=r^{2}$.  Suppose $n' \mid (q-1)$ and $n'=n_{1}n_{2}$, where $n_{1}=\gcd(n', r+1)$ and $n_{2}=\frac{n'}{\gcd(n',r+1)}$. Let $1 \leq t \leq \frac{r-1}{n_{2}}$ and $n=tn'$,
\begin{description}
       \item[(i)] if both $n$ and $\frac{r+1}{n_{1}}$ are even, then there exists a $q$-ary MDS  self-dual code of length $n$;
       \item[(ii)] if $n$ is odd, then there exists a $q$-ary MDS self-dual code of length $n+1$;
       \item[(iii)] if $n$ is even, then there exists a $q$-ary MDS self-dual code of length $n+2$.
     \end{description}
\end{theorem}

\begin{proof}
\textbf{(i):} Let $A$ be defined as \eqref{1}. Then $|A|=n$. For any $1 \leq b \leq t$,
  \[\pi_{A_{b}}(x)=\prod_{e \in A_{b}}(x-e)=\prod_{h \in H}(x-\beta_{b}h)=x^{n'}-\beta_{b}^{n'},\]
  \[\pi_{A_{b}}'(x)=n'x^{n'-1}.\]
  Given $\alpha \in A$, suppose $\alpha=\beta_{b}\omega^{j\frac{q-1}{n'}}$ for some $1 \leq b \leq t$ and $0 \leq j \leq n'-1$. Then by Lemma \ref{lem1},
  \begin{eqnarray*}
    \delta_{A}(\alpha) &=& \delta_{A_{b}}(\alpha)\prod_{s=1, s \neq b}^{t}\pi_{A_{s}}(\alpha) \\
                  &=& n'\alpha^{n'-1}\prod_{s=1, s \neq b}^{t}(\beta_{b}^{n'}-\beta_{s}^{n'}).
  \end{eqnarray*}
  Since $\beta_{b} \in G$, $\beta_{b}=\omega^{\mu_{b}\frac{r+1}{n_{1}}}$ for some $\mu_{b}$. Thus $\beta_{b}$ is a square element of $\mathbb{F}^{*}_{q}$. Since $\frac{q-1}{n'}=\frac{r+1}{n_{1}}\frac{r-1}{n_{2}}$ is even, thus $\alpha$ is a square. Note that for any $1 \leq j \leq t$,
  \[\beta_{j}^{n'}=\omega^{\mu_{j}\frac{r+1}{n_{1}}n'}=(\omega^{\mu_{j}\frac{n'}{n_{1}}})^{r+1} \in \mathbb{F}_{r}.\]
  Hence $\beta_{b}^{n'}-\beta_{s}^{n'} \in \mathbb{F}^{*}_{r}$, which is a square element of $\mathbb{F}^{*}_{q}$. Thus $\delta_{A}(\alpha)$ is a square element of $\mathbb{F}^{*}_{q}$, i.e., $\eta\big(\delta_{A}(\alpha)\big)=1$ for all $\alpha \in A$. The Part (i) then follows from Lemma \ref{lem2}.
  \vskip 2mm \noindent
\textbf{(ii)} and \textbf{(iii):} Let $A_{0}$ be defined as \eqref{2}. Then $|A_{0}|=n+1$. Similar to the proof of Part (i), for any $\alpha \in A$, we have
  \[\delta_{A_{0}}(\alpha)=\alpha\delta_{A}(\alpha)=n'\beta_{b}^{n'}\prod_{s=1, s \neq b}^{t}(\beta_{b}^{n'}-\beta_{s}^{n'}) \in \mathbb{F}_{r}.\]
  Moreover,
  \begin{eqnarray*}
    \delta_{A_{0}}(0) &=& \prod_{s=1}^{t}\pi_{A_{s}}(0) \\
      &=&  \prod_{s=1}^{t}(-\beta_{s}^{n'}) \in \mathbb{F}_{r}.
  \end{eqnarray*}
  Thus for any $e \in A_{0}$, $\eta(\delta_{A_{0}}(e))=\eta(-\delta_{A_{0}}(e))=1$, the conclusions of Part (ii) and Part (iii) then follow from Lemma \ref{lem2} and Lemma \ref{lem3}, respectively.
\end{proof}

\begin{remark}\label{rem1}
In \cite[Theorems 1, 2 and 3]{LLL19}, the authors constructed MDS self-dual codes of length $tn'$, $tn'+1$ and $tn'+2$, respectively, where $n' \mid (q-1)$ and $1 \leq t \leq \frac{r-1}{\gcd(r-1,n')}$. When $n'$ is odd, then $\gcd(r-1,n')\gcd(r+1,n')=n'$ since $n' \mid (r^{2}-1)$. Thus $\gcd(r-1,n')=\frac{n'}{\gcd(r+1,n')}=n_{2}$. The result of our Theorem \ref{thm1} (ii) is equivalent to \cite[Theorem 2]{LLL19}. When both $n'$ and $\frac{q-1}{n'}$ are even, then it can be verified that $\gcd(r-1,n')\gcd(r+1,n')=2n'$. Thus $\gcd(r-1,n')=\frac{2n'}{\gcd(r+1,n')}=2n_{2}$, i.e., $\frac{r-1}{\gcd(r-1,n')}=\frac{1}{2}\frac{r-1}{n_{2}}$.  Thus in this case, our Theorem \ref{thm1} (i) and (iii) are better than \cite[Theorems 1 and 3]{LLL19}, respectively.
\end{remark}

\begin{example}
In Theorem \ref{thm1} (i) and (iii), let $p=r=23$, $q=r^{2}=529$, $n'=12$. Then $n_{2}=\frac{n'}{\gcd(r+1,n')}=1$ and $\frac{r-1}{n_{2}}=22$. Let $t=13$, then $n=tn'=156$. So we can obtain two MDS self-dual codes of length 156 and 158 over $\mathbb{F}_{529}$ from Theorem \ref{thm1} (i) and (iii), respectively. Compared with the parameters of MDS self-dual codes given in Table 1,  these two codes are new.
\end{example}

In the following, we consider another multiplicative subgroup of $\mathbb{F}_{q}$. Suppose $n' \mid (q-1)$. We write $n'=n_{1}n_{2}$, where $n_{1}=\gcd(n', r-1)$ and $n_{2}=\frac{n'}{\gcd(n', r-1)}$. Then $n_{2} \mid (r+1)\frac{r-1}{n_{1}}$. Note that $\gcd(n_{2}, \frac{r-1}{n_{1}})=\gcd(\frac{n'}{n_{1}}, \frac{r-1}{n_{1}})=1$, hence
\[n_{2} \mid (r+1).\]
Let $\omega$ be a primitive element of $\mathbb{F}_{q}$. Denote
\[H=\langle \omega^{\frac{q-1}{n'}} \rangle, G=\langle\omega^{\frac{r-1}{n_{1}}}\rangle.\]
Then $|H|=n'$. Note that $\frac{q-1}{n'}=\frac{r-1}{n_{1}}\cdot\frac{r+1}{n_{2}}$, thus $ \frac{r-1}{n_{1}} \mid \frac{r+1}{n'}$, which deduces that $H$ is a subgroup of $G$.
Then there exist $\beta_{1}, \ldots , \beta_{\frac{r+1}{n_{2}}} \in G$
such that $\{\beta_{b} H\}^{\frac{r+1}{n_{2}}}_{b
=1}$ represent all cosets of $G/H$.

Let $1 \leq t \leq \frac{r+1}{n_{2}}$ and $n=tn'$. Denote $A_{b}=\beta_{b}H$,
\begin{equation}\label{3}
 A=\bigcup^{t}_{b=1}A_{b},
\end{equation}
and
\begin{equation}\label{4}
 A_{0}=A\cup\{0\}.
\end{equation}
Based on \eqref{3} and \eqref{4}, we give our second construction of MDS self-dual codes as follows.
\begin{theorem}\label{thm2}
   Let $q=r^{2}$.  Suppose $n' \mid (q-1)$ and $n'=n_{1}n_{2}$, where $n_{1}=\gcd(n', r-1)$ and $n_{2}=\frac{n'}{\gcd(n',r-1)}$. Let $1 \leq t \leq \frac{r+1}{n_{2}}$ and $n=tn'$,
\begin{description}
       \item[(i)] if both $\frac{r-1}{n_{1}}$ and $tn_{2}$ are even, then there exists a $q$-ary MDS self-dual code of length $n$;
       \item[(ii)] if both $n_{2}$ and $\frac{r+1}{2}(t-1)$ are even, or $n_{2}$ is odd and $t$ is even with $1 \leq t \leq \frac{r+1}{n_{2}}-1$, then there exists a $q$-ary MDS self-dual code of length $n+2$.
     \end{description}
\end{theorem}

\begin{proof}
\textbf{(i):} Let $A$ be defined as \eqref{3}. Then $|A|=n$. For any $1 \leq b \leq t$,
  \[\pi_{A_{b}}(x)=\prod_{e \in A_{b}}(x-e)=\prod_{h \in H}(x-\beta_{b}h)=x^{n'}-\beta_{b}^{n'},\]
  \[\pi_{A_{b}}'(x)=n'x^{n'-1}.\]
  Given $\alpha \in A$, suppose $\alpha=\beta_{b}\omega^{j\frac{q-1}{n'}}$ for some $1 \leq b \leq t$ and $0 \leq j \leq n'-1$. Then by Lemma \ref{lem1},
  \begin{eqnarray*}
    \delta_{A}(\alpha) &=& \delta_{A_{b}}(\alpha)\prod_{s=1, s \neq b}^{t}\pi_{A_{s}}(\alpha) \\
                  &=& n'\alpha^{n'-1}\prod_{s=1, s \neq b}^{t}(\beta_{b}^{n'}-\beta_{s}^{n'}).
  \end{eqnarray*}
  For any $1 \leq j \leq t$, since $\beta_{j} \in G$, we write $\beta_{j}=\omega^{\mu_{j}\frac{r-1}{n_{1}}}$ for some integer $\mu_{j}$. Thus
  \[\beta_{j}^{n'}=(\omega^{\mu_{j}\frac{r-1}{n_{1}}})^{n'}=(\omega^{\mu_{j}\frac{n'}{n_{1}}})^{r-1}.\]
  Hence
  \[\beta_{j}^{n'(r+1)}=1, i.e., \beta_{j}^{n'r}=\beta_{j}^{-n'}.\]
  Denote $\Omega=\prod_{s=1, s \neq b}^{t}(\beta_{b}^{n'}-\beta_{s}^{n'})$. Then
  \begin{eqnarray*}
    \Omega^{r} &=& \prod_{s=1, s \neq b}^{t}(\beta_{b}^{-n'}-\beta_{s}^{-n'}) \\
      &=& \prod_{s=1, s \neq b}^{t}(\beta_{b}\beta_{s})^{-n'}(\beta_{s}^{n'}-\beta_{b}^{n'}) \\
      &=& (-1)^{t-1}\beta_{b}^{-(t-2)n'} \prod_{s=1}^{t}\beta_{s}^{-n'}\Omega.
  \end{eqnarray*}
  Hence
  \begin{eqnarray*}
    \Omega^{r-1} &=& (-1)^{t-1}\beta_{b}^{-(t-2)n'} \prod_{s=1}^{t}\beta_{s}^{-n'} \\
      &=& (-1)^{t-1}\omega^{-\mu_{b}\frac{r-1}{n_{1}}(t-2)n'}\prod_{s=1}^{t}\omega^{-\mu_{s}\frac{r-1}{n_{1}}n'} \\
     &=& (-1)^{t-1}\omega^{-n_{2}(r-1)(\mu_{b}(t-2)+\sum_{s=1}^{t}\mu_{s})}.
  \end{eqnarray*}
  Note that $-1=\omega^{\frac{r+1}{2}(r-1)}$. Thus there exists an integer $k$ such that
  \[\Omega=\omega^{\frac{r+1}{2}(t-1)-n_{2}(\mu_{b}(t-2)+\sum_{s=1}^{t}\mu_{s})}.\]
   Note that $\alpha=\beta_{b}\omega^{j\frac{q-1}{n'}}=\omega^{\mu_{b}\frac{r-1}{n_{1}}+j\frac{q-1}{n'}}$ which is a square since both $\frac{r-1}{n_{1}}$ and $\frac{q-1}{n'}=\frac{r-1}{n_{1}}\frac{r+1}{n_{2}}$ are even. Thus $\eta\big(\delta_{A}(\alpha)\big)=\eta(\omega^{\frac{r+1}{2}(t-1)-n_{2}(\sum_{s=1}^{t}\mu_{s})})$ for all $\alpha \in A$. The Part (i) then follows from Lemma \ref{lem2}.
  \vskip 2mm \noindent
\textbf{(ii):} Let $A_{0}$ be defined as \eqref{4}. Then $|A_{0}|=n+1$. According to the proof of Part (i), for any $\alpha=\beta_{b}\omega^{j\frac{q-1}{n'}} \in A$, where $1 \leq b \leq t$ and $0 \leq j \leq n'-1$, we have
\begin{eqnarray*}
  \delta_{A_{0}}(\alpha) &=& \alpha\delta_{A}(\alpha) \\
    &=& n'\beta_{b}^{n'}\prod_{s=1, s \neq b}^{t}(\beta_{b}^{n'}-\beta_{s}^{n'}) \\
    &=& n'\omega^{\mu_{b}(r-1)n_{2}+\frac{r+1}{2}(t-1)-n_{2}(\mu_{b}(t-2)+\sum_{s=1}^{t}\mu_{s})},
\end{eqnarray*}
for some integer $k$. If both $n_{2}$ and $\frac{r+1}{2}(t-1)$ are even, then $\delta_{A_{0}}(\alpha)$ is a square in $\mathbb{F}^{*}_{q}$, i.e., $\eta(\delta_{A_{0}}(\alpha))=1$. If $n_{2}$ is odd and $t$ is even, then $\eta(\delta_{A_{0}}(\alpha))=\eta(\omega^{\frac{r+1}{2}+\sum_{s=1}^{t}\mu_{s}}).$ Since $1 \leq \mu_{1}, \mu_{2}, \dots, \mu_{t} \leq \frac{r+1}{n_{2}}$ and $1 \leq t \leq \frac{r+1}{n_{2}}-1$, we can always choose suitable $\mu_{1}, \mu_{2}, \dots, \mu_{t}$ such that $\frac{r+1}{2}+\sum_{s=1}^{t}\mu_{s}$ is even. Thus  $\eta(\delta_{A_{0}}(\alpha))=1$.
\vskip 2mm \noindent
Moreover,
  \begin{eqnarray*}
    \delta_{A_{0}}(0) &=& \prod_{s=1}^{t}\pi_{A_{s}}(0) \\
      &=&  \prod_{s=1}^{t}(-\beta_{s}^{n'}) \\
      &=& (-1)^{t}\omega^{n_{2}(r-1)\sum_{s=1}^{t}\mu_{s}},
  \end{eqnarray*}
  which is a square in $\mathbb{F}^{*}_{q}$. In summary, for any $e \in A_{0}$, $\eta(\delta_{A_{0}}(e))=\eta(-\delta_{A_{0}}(e))=1$. Part (ii) then follows from Lemma \ref{lem3}.
\end{proof}

\begin{remark}
In \cite[Theorems 1 and 3]{FLLL19}, the authors constructed MDS self-dual codes of length $tn'$ and $tn'+2$, respectively, where $n' \mid (q-1)$ and $1 \leq t \leq \frac{r+1}{\gcd(r+1,n')}$. Similar to Remark \ref{rem1}, when both $n'$ and $\frac{q-1}{n'}$ are even, then it can be verified that $\gcd(r-1,n')\gcd(r+1,n')=2n'$. Thus $\gcd(r+1,n')=\frac{2n'}{\gcd(r-1,n')}=2n_{2}$, i.e., $\frac{r+1}{\gcd(r+1,n')}=\frac{1}{2}\frac{r+1}{n_{2}}$.  Thus in this case, our Theorem \ref{thm2} (i) and (ii) are better than \cite[Theorems 1 and 2]{FLLL19}, respectively.
\end{remark}

\begin{example}

\begin{description}
  \item[(i)] In Theorem \ref{thm2}, let $r=25$ and $q=r^{2}=625$. Let $n'=12$ then $n_{2}=\frac{n'}{\gcd(n', r-1)}=1$ and $\frac{r+1}{n_{2}}=26$. Let $t=14$, then $n=tn'=168$, thus we can obtain an MDS self-dual code of length 168 over $\mathbb{F}_{625}$ from Theorem \ref{thm2} (i).
  \item[(ii)] In Theorem \ref{thm2}, let $r=19$ and $q=r^{2}=361$. Let $n'=12$ then $n_{2}=\frac{n'}{\gcd(n', r-1)}=2$ and $\frac{r+1}{n_{2}}=10$. Let $t=7$, then $\frac{r+1}{2}(t-1)$ is even and $n=tn'=84$, thus we can obtain an MDS self-dual code of length 86 over $\mathbb{F}_{361}$ from Theorem \ref{thm2} (ii).
\end{description}
Compared with the parameters of MDS self-dual codes given in Table 1,  these two codes are new.
\end{example}

\section{Self-Dual GRS Codes with Different Evaluation Sets}
In this section, we give a systematic way to construct self-dual GRS codes provided any known self-dual GRS code.

 Let $\mathbb{P}^{k}(\mathbb{F}_{q})$ be the projective geometry over $\mathbb{F}_{q}$ of dimension $k$ and let  $PGL_{k+1}(\mathbb{F}_{q})$ be the group of automorphisms of $\mathbb{P}^{k}(\mathbb{F}_{q})$.
\begin{definition}[M{\"o}bius action]
	For any $g= (\bbsm a & b \\ c & d \besm) \in PGL_2(\mathbb{F}_{q})$, the M{\"o}bius action on $\mathbb{P}^{1}(\mathbb{F}_{q})\simeq\mathbb{F}_{q}\cup\infty$ is defined to be $$g(t)=\frac{c+dt}{a+bt}=\begin{cases}
	\frac{c+dt}{a+bt},&\mbox{if $t\in \mathbb{F}_{q}\setminus\{-\frac{a}{b}\}$};\\
	\infty,&\mbox{if $t=-\frac{a}{b}$;}\\
	\frac{d}{b},&\mbox{if $t=\infty.$}
	\end{cases}$$
\end{definition}

\begin{definition} \label{def_gm}
	For each $2 \leq k \leq q$,  we define the map
	$GL_2(\mathbb{F}_{q}) \to  GL_k(\mathbb{F}_{q})$ denoted $g \mapsto g_k$  as follows. For $g = (\bbsm a & b \\ c & d \besm) \in GL_2(\mathbb{F}_{q})$,  the
	$(i,j)$-th ($1 \leq i,j \leq k$) entry of $g_k$  is the coefficient of $X^{j-1}$ in the polynomial $(a+b X)^{k-i} (c+d X)^{i-1}$.
\end{definition}

Recall that for any $\alpha\in \mathbb{F}_{q}\cup\infty$, $\bm c_k(\alpha)$ is defined as
\[
\bm c_k(\alpha)\triangleq\begin{cases}
(1,\alpha,\alpha^2,\cdots, \alpha^{k-1})^T\in\mathbb{F}_{q}^k&\mbox{if $\alpha\in\mathbb{F}_{q}$;}\\
(0,0,\cdots,0,1)^T\in\mathbb{F}_{q}^k&\mbox{if $\alpha=\infty$.}\\
\end{cases}
\]
We collect some properties of the matrix $g_k$ (see \cite{BGHK17}):
\begin{lemma}[{\cite[Propositions 2.6 and 2.5]{BGHK17}}]
	Keep the notations as above. Then
	\begin{enumerate}
		\item  The map $g \mapsto g_k$ is a group homomorphism  and the induced homomorphism $PGL_2(\mathbb{F}_{q}) \to PGL_k(\mathbb{F}_{q})$ (which we again denote by $g \mapsto g_k$) is a monomorphism.
		\item  For each $t \in \mathbb{F}_{q} \cup \infty$ we have
		\begin{equation*} \label{eq:aut_rnc} g_k \textbf{\emph{c}}_k(t) = \textbf{c}_k(g(t)) \in \mathbb{P}^{k-1}(\mathbb{F}_{q}). \end{equation*}
	\end{enumerate}
\end{lemma}

The action of $PGL_2(\mathbb{F}_{q})$ on $\mathbb{F}_{q} \cup \infty=\{\alpha_1,\alpha_2,\cdots,\alpha_{q+1}=\infty\}$  gives a monomorphism $g \mapsto \Pi(g)$ from $PGL_2(\mathbb{F}_{q})$ to the group of permutation matrices in $GL_{q+1}(\mathbb{F}_{q})$ defined by:
\[ (\alpha_1, \dots,  \alpha_{q+1}) \Pi(g)\triangleq (g \alpha_1, \dots, g \alpha_{q+1}).\]
For any $g= (\bbsm a & b \\ c & d \besm) \in PGL_2(\mathbb{F}_{q})$, define a diagonal matrix $\Delta_k(g)=\text{diag}(\delta_1, \delta_2, \cdots, \delta_{q+1})$ as follows:
\begin{align*}
\delta_i=\begin{cases}
(a+b\alpha_i)^{k-1} &\mbox { if $\alpha_i\neq -\frac{a}{b}, \infty$},\\
(c-d\frac{a}{b})^{k-1} &\mbox{ if  $b \neq 0$, $\alpha_i= -\frac{a}{b}$},\\
b^{k-1}  &\mbox{ if $b \neq 0$, $\alpha_i= \infty$},\\
d^{k-1}  &\mbox{ if $b=0$, $\alpha_i= \infty$}.\\
\end{cases}
\end{align*}
It is easy to see that $\delta_i \in \mathbb{F}_{q}^{*}$.
Let $\textbf{1}$ represent the all $1$ vector. For any $2\leq k\leq q,$ denote the generator matrix of $GRS_k(\mathbb{F}_{q}\cup \infty, \textbf{1})$ by
	\begin{align*}
G_k=[c_k(\alpha)]_{\alpha\in \mathbb{F}_{q}\cup\infty}=\left(
\begin{array}{cccc}
1 &  \cdots & 1 & 0 \\
\alpha_1  & \cdots & \alpha_q & 0 \\
\vdots  & \ddots & \vdots & \vdots \\
\alpha_1^{k-2}  & \cdots & \alpha_q^{k-2} & 0 \\
\alpha_1^{k-1}  & \cdots & \alpha_q^{k-1} & 1 \\
\end{array}
\right).
\end{align*}
\begin{theorem}[\cite{BGHK17}]\label{automorphisms}
For any $g\in PGL_2(\mathbb{F}_{q})$, $g_k$ is an automorphism of $GRS_k(\mathbb{F}_{q}\cup \infty, \textbf{1})$. More precisely, we have
\[
g_kG_k=G_k\Pi(g)\Delta_k(g).
\]

\end{theorem}
Now, we provide our main result of this section as follows.
\begin{theorem}
Suppose $A$ is a subset $A\subset \mathbb{F}_q\cup \infty$ of size $n$ $(n~is~even)$, such that the generalized RS code $GRS_k(A,\bm v)$ is self-dual $(k=\frac{n}{2})$ for some scaling vector $\bm v=(v_1,\cdots,v_n)\in(\mathbb{F}_q^*)^n$. For any $g \in PGL_2(\mathbb{F}_{q}),$ let $A'=gA=\{g(\alpha)\,:\,\alpha\in A\},$ then there exists a scaling vector $\bm v'\in (\mathbb{F}_{q}^*)^{n}$ such that $GRS_k(A',\bm v')$ is self-dual.
\end{theorem}
\begin{proof}
 Let $G'=[c_k(g\alpha_1),c_k(g\alpha_2),\cdots, c_k(g\alpha_{q+1})].$ By the definition of $\Pi(g)$, we have
 \[G'=G_{k}\Pi(g).\]
By Theorem~\ref{automorphisms}, there exists $g_k\in GL_k(\mathbb{F}_{q})$ such that
\begin{equation}\label{5}
  g_kG_k=G'\Delta_k(g).
\end{equation}
Since $g_k$ is an automorphism of $GRS_k(\mathbb{F}_{q}\cup \infty, \textbf{1})$, $g_kG_k$ is a new generator matrix of $GRS_k(\mathbb{F}_{q}\cup \infty,\textbf{1})$.

Let $\mathfrak{A}$ be the subset of $\{1,2,\cdots, q+1\}$ such that $i\in \mathfrak{A}$ if and only if $\alpha_i\in A.$ For any $k \times (q+1)$-matrix $M$, we denote $\mbox{Columns}(M, \mathfrak{A})$ the sub-matrix of $M$ whose columns are the columns of $M$ indexed by $\mathfrak{A}$. Now consider the restriction of both sides of Eq. \eqref{5} to columns indexed by $\mathfrak{A}$. We have
\[\mbox{Columns}(g_kG_k, \mathfrak{A})=\mbox{Columns}(G'\Delta_k(g), \mathfrak{A}).\]
Hence
	\[
	\mbox{Columns}(g_kG_k, \mathfrak{A})\mbox{diag}\{v_1,\cdots, v_{n}\}=\mbox{Columns}(G'\Delta_k(g), \mathfrak{A})\mbox{diag}\{v_1,\cdots, v_{n}\}.
	\]
Note that  $\mbox{Columns}(g_kG_k, \mathfrak{A})\mbox{diag}\{v_1,\cdots, v_{n}\}$ is a generator matrix of $GRS_k(A,\bm v)$. It is easy to verify that $\mbox{Columns}(G'\Delta_k(g), \mathfrak{A})=\mbox{Columns}(G', \mathfrak{A})\mbox{diag}\{\delta_i\,:\,i\in \mathfrak{A}\}$. Hence
\[\mbox{Columns}(G'\Delta_k(g), \mathfrak{A})\mbox{diag}\{v_1,\cdots, v_{n}\}=\mbox{Columns}(G', \mathfrak{A})\mbox{diag}\{\delta_i\,:\,i\in \mathfrak{A}\}\mbox{diag}\{v_1,\cdots, v_{n}\}.\]
The right hand of the above equation is a generator matrix of the GRS code with the evaluation set $A'=\{g(\alpha_i)\,|\,i \in \mathcal{A}\}$ and scaling vector $v'=\mbox{diag}\{\delta_i\,:\,i\in \mathfrak{S}\}\cdot(v_1,\cdots, v_{n})^{T}$.
	So $GRS_k(A,\bm v)=GRS_k(A',\bm v').$ Hence, $GRS_k(A',\bm v')$ is a self-dual code.
\end{proof}	
\begin{remark}
The above theorem provides various options for the evaluation sets to obtain self-dual GRS codes.
\end{remark}

As a corollary,  all the self-dual extended GRS codes over finite field $\mathbb{F}_{q}$ with length $n< q+1$ can be constructed from GRS codes with the same parameters.
\begin{corollary}
	Suppose $A$ is a subset $A\subset \mathbb{F}_q\cup \infty$  of size $n$ with $\infty\in A$ and $n$ is even, such that the extended generalized RS code $GRS_k(A,\bm v)$ is self-dual $(k=\frac{n}{2})$ for some scaling vector $\bm v=(v_1,\cdots,v_n)\in(\mathbb{F}_q^*)^n$. Then there exist a subset $A'\subset \mathbb{F}_q$ and scaling vector $\bm v'\in (\mathbb{F}_{q}^*)^{n}$ such that $GRS_k(A',\bm v')$ is self-dual. In other words, one can choose $A'$ to construct self-dual GRS code instead of $A$ containing the infinity point.
\end{corollary}
\begin{proof}
Choose any $-a\in \mathbb{F}_q\setminus A$ and let $g= (\bbsm a & 1 \\ 1 & 0 \besm) \in PGL_2(\mathbb{F}_{q})$. Thus $A'=gA=\{\frac{1}{\alpha+a}\,:\,\alpha\in A\}\subset \mathbb{F}_{q}.$ The corollary then follows from Theorem \ref{automorphisms}.
\end{proof}

\section{Conclusion}
In this paper, we investigate the construction of MDS self-dual codes via GRS codes. First, we present four families of self-dual GRS codes via some multiplicative subgroups of finite fields. These MDS self-dual codes are new in the sense that their parameters can not be covered in the literature. Second, we give a systematic approach to obtain self-dual GRS codes from any known self-dual GRS code via a particular family of automorphisms of GRS codes. In particular, we prove that any $q$-ary self-dual extended GRS codes of length $n <q+1$ can be obtained from GRS codes with the same parameters.










\end{document}